%% file: main.tex
\documentclass[letterpaper]{article}
\usepackage{aaai}
\usepackage{times}
\usepackage{helvet} 
\usepackage{courier} 
\usepackage{color}
\usepackage{latexsym}
\usepackage{amsmath}
\usepackage{amssymb}
\usepackage{mathptmx}
\usepackage{enumerate}
\usepackage{amsfonts}
\usepackage{marvosym}
\usepackage{verbatim}

\usepackage{amsthm}
\usepackage{xspace}

\newtheorem{thm}{Theorem}

\newtheorem{lem}[thm]{Lemma}

\newtheorem{defn}[thm]{Definition}

\newtheorem{observation}{Observation}

\newcommand{\non}{\mathnormal{\sim}}
\newcommand{\To}{\Rightarrow}
\newcommand{\TO}{\rightarrow}
\newcommand{\defeater}{\leadsto}

\newcommand{\set}[2][\relax]{\ensuremath{#1\{#2#1\}}}
\newcommand{\seq}[2][n]{#2_1,\dots,#2_{#1}}
\newcommand{\ds}{\displaystyle}

\newcommand{\LAB}{\ensuremath{\mathrm{Lbl}}\xspace}
\newcommand{\FACTS}{\ensuremath{\mathrm{F}}\xspace}
\newcommand{\PROP}{\ensuremath{\mathrm{PROP}}\xspace}
\newcommand{\LIT}{\ensuremath{\mathrm{Lit}}\xspace}
\newcommand{\MOD}{\ensuremath{\mathrm{MOD}}\xspace}
\newcommand{\MODLIT}{\ensuremath{\mathrm{ModLit}}\xspace}

\newcommand{\pr}{\ensuremath{\mathsf{Pr}}}
\newcommand{\op}{\ensuremath{\mathsf{Op}}}
\newcommand{\com}{\ensuremath{\mathsf{Com}}}

\newcommand{\Arg}{\ensuremath{\mathsf{arg}}}
\newcommand{\soc}{\ensuremath{\mathsf{soc}}}

\newcommand{\OBL}{\ensuremath{\mathsf{OBL}}\xspace}
\newcommand{\BEL}{\ensuremath{\mathsf{BEL}}\xspace}
\newcommand{\INT}{\ensuremath{\mathsf{INT}}\xspace}

\newcommand{\Convert}{\ensuremath{\mathrm{Cv}}\xspace}
\newcommand{\Conflict}{\ensuremath{\mathrm{Cf}}\xspace}
\newcommand{\Aconf}[2]{\mathrm{Cf}(#1,#2)}
\newcommand{\Aconv}[2]{\mathrm{Cv}(#1,#2)}

\newcommand{\pflat}{\ensuremath{\mathrm{pflat}}\xspace}
\newcommand{\fflat}{\ensuremath{\mathrm{flat}}\xspace}
\newcommand{\Pflat}[1]{\mathrm{pflat}(#1)}
\newcommand{\Flatulo}[1]{\mathrm{flat}(#1)}
\newcommand{\Flat}[2]{\mathrm{flat}(#1#2)}
\newcommand{\flatx}[2]{#1\_\mathrm{pflat}(#2)}

\newcommand{\str}{\ensuremath{\mathrm{s}}\xspace}
\newcommand{\sd}{\ensuremath{\mathrm{sd}}\xspace}
\newcommand{\de}{\ensuremath{\mathrm{d}}\xspace}
\newcommand{\dft}{\ensuremath{\mathrm{dft}}\xspace}

\title{Strategic Argumentation is NP-Complete}
\author{Guido Governatori\\
NICTA, Australia\\
\And Francesco Olivieri\\
Griffith Univ., Australia\\
NICTA, Australia\\
Verona Univ., Italy
\And Simone Scannapieco\\
Griffith Univ., Australia\\
NICTA, Australia\\
Verona Univ., Italy
\And Antonino Rotolo\\
Bologna Univ., Italy
\And Matteo Cristani\\
Verona Univ., Italy}
\begin{document}

\maketitle
\begin{abstract}
In this paper we study the complexity of strategic argumentation for dialogue games. A dialogue game is a 2-player game where the parties play arguments. We show how to model dialogue games in a skeptical, non-monotonic formalism, and we show that the problem of deciding what move (set of rules) to play at each turn is an NP-complete problem.
\end{abstract}

\input{Intro}

\input{Logic}

\input{DialogueGames}

\input{Reduction}


\bibliography{biblio,biblio-NINO}
\bibliographystyle{aaai}
\end{document}

%% file: Intro.tex
\section{Introduction and Motivation}\label{sec:introduction}

Over the years many dialogue games for argumentation have been proposed to
study questions such as which conclusions are justified, or how procedures for
debate and conflict resolution should be structured to arrive at a fair and
just outcome. We observed that the outcome of a debate does not solely
depend on the premises of a case, but also on the strategies that parties in a
dispute actually adopt. According to our knowledge, this aspect has not received the proper attention in the literature of the field.

Almost all the AI literature on the strategic aspects of argumentation (see
Section \ref{sec:related_work} for a brief overview) assumes to work with
argument games with complete information, i.e., dialogues where the structure of
the game is common knowledge among the players. Consider, however, the
following example due to \cite{Sato:2011} (which in turn modifies an example
taken from \cite{OkunoT09}):
\[
\begin{array}{ll}
p_0: & \text{``You killed the victim.''} \\
c_1: & \text{``I did not commit murder! There is no evidence!''}\\
p_1: & \text{``There is evidence. We found your ID card near the} \\
& \text{scene.''}\\
c_2: & \text{``It's not evidence! I had my ID card stolen!''}\\
p_2: & \text{``It is you who killed the victim. Only you were near}\\
& \text{the scene at the time of the murder.''}\\
c_3: & \text{``I didn't go there. I was at facility A at that time.''}\\
p_3: & \text{``At facility A? Then, it's impossible to have had }\\
& \text{your ID card stolen since
facility A does not allow a} \\
& \text{person to enter without an ID card.''}
\end{array}
\]

This dialogue exemplifies an argument game occurring in witness examinations in legal courts. The peculiarity of this game is the fact that the exchange of arguments reflects an asymmetry of information between the players: each player does not know the other player's knowledge, thus she cannot predict which arguments are attacked and which counterarguments are employed for attacking the arguments. Indeed, \cite{Sato:2011} points out, for instance, that $p_3$ attacks $c_2$, but only when $c_3$ is given: hence, the attack $p_3$ of the proponent is made possible only when the opponent discloses some private information with the move $c_3$.  

Despite the encouraging results offered by \cite{Sato:2011}, we argue that 
relaxing the complete-information assumption leads in general to non-tractable frameworks. In this paper, in particular, we explore the computational cost of argument games of incomplete information where the (internal) logical structure of arguments is considered. 

In this case relaxing complete information, such as when players do not share the same beliefs and set of arguments, simply amounts to the fact they have different logical theories, i.e., different sets of rules from which arguments supporting logical conclusions can be built. Hence, if the proponent, having a theory $T$, has the objective to prove that some $l$ is true, there is no obvious way for preferring an argument for $l$ obtained from the minimal subset of $T$ (which could at first sight minimise the chances of successful attacks from the opponent) over the maximal set of arguments obtained from the whole $T$ (which could at a first sight maximise the chances to defeat any counterarguments of the opponent).

The layout of the paper is as follows. Section \ref{sec:intuition} offers a
gentle introduction and motivation for our research problem. Section
\ref{sec:related_work} reviews relevant related work, thus presenting further
motivations behind our contribution. Section \ref{sec:logic} presents the logic
used for building arguments in dialogues (Argumentation Logic): it is a variant
of Defeasible Logic \cite{tocl} having linear complexity; another logic (Agent
Logic has linear complexity as well) is subsequently recalled from
\cite{jaamas:bio}. In this second logic, it is possible to formulate the
NP-complete ``Restoring Sociality Problem''; the objective is to prove that
this problem can be mapped into the problem of interest here, the so-called
``Strategic Argumentation Problem'', which rather consists in successfully
deciding for each player what move to play at each argument turn (thus showing
that the Strategic Argumentation Problem is NP-complete as well). Section
\ref{sec:dialogue_games} defines dialogue protocols for games of incomplete
information based on Argumentation Logic and formulates the Strategic
Argumentation Problem. Section \ref{sec:reduction} shows how to transform a
theory in Agent Logic into an equivalent one in Argumentation Logic, and
presents the main theorem of computational complexity for argument games.

\section{A Gentle Introduction to the Problem}\label{sec:intuition}
%
In the most typical forms of strategic argumentation, two players exchange arguments in a dialogue game: in the simplest case, a proponent (hereafter $\pr$) has the objective to prove a conclusion $l$ (a literal of the language) and an opponent (hereafter $\op$) presents counterarguments to the moves of $\pr$. If we frame this intuition in proof-theoretic settings, such as in those developed in \cite{GovernatoriMAB04,Prakken10,Toni13} where arguments are defined as inference trees formed by applying rules, exchanging arguments means exchanging logical theories (consisting of rules) proving conclusions. Assume, for instance, that the argument game is based on a finite set $F$ of indisputable facts and a finite set $R$ of rules: facts initially fire rules and this leads to building proofs for literals. 
%

If $R$ and $F$ are common knowledge of $\pr$ and $\op$, successful strategies in argument games are trivially identified: each player can compute if the entire theory (consisting of $F$ and $R$) logically entails $l$. In this situation the game consists of a single move. 

Suppose now that $F$ is known by both players, but $R$ is partitioned into three subsets: a set $R_{\com}$ known by both players and two subsets $R_{\pr}$ and $R_{\op}$ corresponding, respectively, to $\pr$'s and $\op$'s private knowledge (what $\pr$ and $\op$ privately know to be true). This scenario exemplifies an argument game of incomplete information. 
%
In this context, each player can use all rules belonging to her private knowledge ($R_{\pr}$ or $R_{\op}$) as well as all the public rules. These rules are not just the rules in $R_{\com}$ but also rules that, though initially belonging to the private information of other player, have been used in previous turns. 

Let us suppose to work with a skeptical non-monotonic framework, i.e., a logical machinery where, whenever two conflicting conclusions are obtainable from different rules, the system refrains to take a decision. Assuming a game where players has private and public knowledge, the problem of deciding what move (set of rules) to play at each turn amounts to establish whether there is any subset of her rules that can be successful. Is there any safe criterion to select successful strategies?

Consider the following three examples.

$\pr$ and $\op$ are debating about the truthfulness of a statement, we say $l$; $\pr$ is arguing that $l$ is the case, whilst $\op$ answers back the truthfulness of the opposite claim (henceforth $\neg l$). Each player has her own (private) arguments, not known by the opponent, but they both share the factual knowledge as well as some inference rules. Suppose $\pr$ has the following private arguments:

\[
\begin{array}{lccccccc}
	\mathcal{P}_{1}: & a & \To & b & \To & c & \To & l\\
	\mathcal{P}_{2}: & \neg b & \To & \neg e & \To & f & & \\
	\mathcal{P}_{3}: & \neg b & \To & \neg e & \To & g & & \\
	\mathcal{P}_{4}: & d & \To & c, &   &   &   &  \\	
\end{array}
\]
while $\op$ has
\[
\begin{array}{lccccc}
	\mathcal{O}_{1}: & a & \To & e & \To & \neg l\\
	\mathcal{O}_{2}: & d & \To & \neg b &  & \\
	\mathcal{O}_{3}: & f & \To & \neg l, & &\\
	
\end{array}
\]
where $F = \set{a, d}$ and $R_{\com}= \set{g \To \neg l}$. The notation used is to exemplify arguments as chains of rules. For instance, argument $\mathcal{P}_{1}$ implies that $R_{\pr}$ contains three rules $r_{1}: a \To b$, $r_{2}: b \To c$, and $r_{3}: c \To l$.

The point of the example being that if $\pr$ decides to announce all his
private arguments, then she is not able to prove her thesis $l$. Indeed, she
would not have counterarguments defeating $\mathcal{O}_{3}$ and
$R_{\com}$. If instead she argues with $\mathcal{P}_{1}$ and the subpart $\neg b \To \neg e$ of $\mathcal{P}_{2}$, keeping hidden from $\op$
the way to prove the premises $d$ of $\mathcal{O}_{2}$, then she proves $l$

Consider now this new setting:
\begin{align*}
    F  & = \{  a, d, f\} \\
    R_{\com} & =  \emptyset \\
    R_{\pr} & = \{  a\To b ,\quad  
                     d\To c ,\quad
                      c \To b \}\\
    R_{\op} & = \{  c\To e ,\quad
                        e,f\To \neg b \}
\end{align*} 
If $\pr$'s intent is to prove $b$ and she plays $\set{a\To b}$, then $\pr$ wins the game. However, if $\pr$ plays $\set{d\To c,\; c\To b}$ (or even $R_{\pr}$), this allows $\op$ to succeed. Here, a minimal subset of $R_{\pr}$ is successful. However, the situation (for similar reasons) can be reversed for $\pr$:
\begin{align*}
    F & = \{  a, d, f\} \\
    R_{\com} & =  \emptyset \\
    R_{\pr} & = \{  a\To b ,\quad
                   d\To \neg c \}\\
    R_{\op} &= \{d, c\To \neg b ,\quad
                       f\To c \}
\end{align*}
In this second case, the move $\set{a\To b}$ is not successful for $\pr$, while playing with the whole $R_{\pr}$ ensures victory.

In the remainder of this paper, we will study this research question in the 
context of Defeasible Logic.  We will show that the problem of deciding what 
set of rules to play (Strategic Argumentation Problem) at a given move 
is NP-complete even when the problem of deciding whether a given theory 
(defeasibly) entails a literal can be computed in polynomial time. 
We will map the NP-complete Restoring Sociality Problem proposed in
\cite{jaamas:bio} into the Strategic Argumentation Problem. To this end, we
first propose a standard Defeasible Logic to formalise the argumentation
framework (Subsection~{\ref{subsec:argumentation_logic}) and then we present
the BIO agent defeasible logic (Subsection~\ref{subsec:agent_logic}). Finally,
in Section~\ref{sec:reduction} we show how to transform an agent defeasible
logic into an equivalent argumentation one and we present the main theorem of
computational complexity.




\section{Related Work}\label{sec:related_work}

Despite the game-like character of arguments and debates, game-theoretic investigations of argumentation are still rare in the AI argumentation literature and in the game theory one as well (an exception in this second perspective is \cite{GlazerR01}). 

Most existing game-theoretic investigations of argumentation in AI, such as \cite{Procaccia:2005,toni-jelia:2008,COMMA:2008,Rahwan:2009,Grossi:2013} proceed within Dung's abstract argumentation paradigm, while \cite{ICAIL:2007}, though working on argumentation semantics related with Dung's approach, develop a framework where also the logical internal structure of arguments is made explicit.

\cite{toni-jelia:2008} presents a notion of argument strength within the class of games of strategy. The measure of the strength of an argument emerges from confronting
proponent and opponent via a repeated game of argumentation strategy such that the payoffs reflect the long term interaction between  proponent and opponent strategies.

Other types of game analyses have been used for argumentation. In particular, argumentation games have been reconstructed as two-player extensive-form games of perfect information \cite{Procaccia:2005,COMMA:2008,Grossi:2013}.  (For a discussion on using extensive-form games, see also \cite{Rahwan:2009}.) While \cite{Grossi:2013} works on zero-sum games, \cite{COMMA:2008} does not adopt this view because preferences over outcomes are specified in terms of expected utility combining the probability of success of arguments (with respect to a third party, an adjudicator such as a judge) with the costs and benefits associated to arguments, thus making possible that argument withdrawn be the most preferred option. Besides this difference, in both approaches uncertainty is introduced  due to different probabilities of success depending on a third party, such as external audience or a judge, whose attitude towards the arguments exchanged by proponent and opponent is uncertain. 

All these works assume that argument games have complete information, which, we noticed, is an oversimplification is many real-life contexts (such as in legal disputes). How to go beyond complete information? In game-theoretic terms, one of the simplest ways of analyzing argument games of incomplete information is to frame them as 
\emph{Bayesian extensive games with observable actions} \cite[chap.~12]{Osborne99}: 
this is possible because every player observes the argumentative move of the other player and uncertainty only derives from an initial move of chance that distributes (payoff-relevant) private information among the players corresponding to logical theories: hence, chance selects types for the players by assigning to them possibly different theories from the set of all possible theories constructible from a given language. If this hypothesis is correct, notice that (i) Bayesian extensive games with observable actions allow to simply extend the argumentation models proposed, e.g., in \cite{COMMA:2008,Grossi:2013}, and  
(ii) the probability distributions over players' types can lead to directly measuring the probability of justification for arguments and conclusions, even when arguments are internally analyzed \cite{AILaw:2012}. Despite this fact, however, complexity results for Bayesian games are far from encouraging (see \cite{GottlobGM07a} for games of strategy). If we move to Bayesian extensive games with observable actions things not encouraging, too. Indeed,  we guess that considerations similar to those presented by \cite{ChalkiadakisB07} can be applied to argument games: the calculation of the perfect Bayesian equilibrium solution can be tremendously complex due to both the size of the strategy space (as a function of the size of the game tree, and it can be computationally hard to compute it \cite{DimopoulosNT02}), and the dependence between variables representing strategies and players' beliefs. A study of these game-theoretical issues cannot be developed here and is left to future work: this paper, instead, considers a more basic question: the computational problem of exploring solutions in the logical space of strategies when arguments have an internal structure. 

In this sense, this contribution does not directly develop any game-theoretic analysis of argumentation games of incomplete information, but it offers results about the computation cost for logically characterizing the problems that any argumentation game with incomplete information potentially rises. Relevant recent papers that studied argumentation of incomplete information without any direct game-theoretic analysis are \cite{OkunoT09} and \cite{Sato:2011}, which worked within the paradigm of abstract argumentation. The general idea in these works is to devise a system for dynamic argumentation games where agents' knowledge bases can change and where such changes are precisely caused by exchanging arguments. \cite{OkunoT09} presents a first version of the framework and an algorithm, for which the authors prove a termination result. \cite{Sato:2011} generalizes this framework (by relaxing some constraints) and devises a computational method to decide which arguments are accepted
by translating argumentation framework into logic programming; this further result, however, is possible only when players are eager to give all the arguments, i.e., when proponent and opponent eventually give all possible arguments in the game.

%% file: Logic.tex
\section{Logic} 
\label{sec:logic}

In this section we shall introduce the two logics used in this paper. The first
is the logic used in a dialogue game. This is the logic to represent the
knowledge of the players, the structure of the arguments, and perform
reasoning. We call this logic ``Argumentation logic'' and we use the Defeasible
Logic of \cite{tocl}. \cite{GovernatoriMAB04} provides the relationships
between this logic (and some of its variants) and abstract argumentation, and
\cite{ThakurGPL07} shows how to use this logic for dialogue games. The second
logic, called Agent Logic, is the logic in which the ``restoring sociality
problem'' (a known NP-completed problem) \cite{jaamas:bio} was formulated. It
is included in this paper to show how to reduce the restoring sociality problem
into the strategic argumentation problem, proving thus that the later is also
an NP-complete problem. The Agent Logic is an extension of Defeasible Logic
with modal operators for Beliefs, Intentions and Obligations \cite{jaamas:bio}.

Admittedly, this section takes a large part of this paper, but is required to
let the reader comprehend the mechanisms behind our demonstration of
NP-completeness.

\subsection{Argumentation Logic} 
\label{subsec:argumentation_logic}

A defeasible argumentation theory is a standard defeasible theory consisting of
a set of \emph{facts} or indisputable statements, a set of rules, and a
\emph{superiority relation} $>$ among rules saying when a single rule may
override the conclusion of another rule. We have that $\phi_{1},\dots ,\phi_{n}
\TO \psi$ is a {\em strict rule} such that whenever the premises
$\phi_{1},\dots ,\phi_{n}$ are indisputable so is the conclusion $\psi$. A
\emph{defeasible rule} $\phi_{1},\dots ,\phi_{n}\To \psi$ is a rule that can be
defeated by contrary evidence. Finally, $\phi_{1},\dots ,\phi_{n}\defeater
\psi$ is a \emph{defeater} that is used to prevent some  conclusion but cannot be used to draw any conclusion.

\begin{defn}[Language]\label{def:languageArg} 
	Let $\PROP$ be a set of
propositional atoms and $\LAB_{arg}$ be a set of labels. Define:
\begin{description}
\item[\emph{Literals}]
\[
\LIT=\PROP\cup \{ \neg p | p\in\PROP\}
\]
If $q$ is a literal, $\non q$ denotes the complementary literal (if
$q$ is a positive literal $p$ then $\non q$ is $\neg p$; and if $q$ is
$\neg p$, then $\non q$ is $p$);
\item[\emph{Rules}]
\[
r: \seq{\phi}\hookrightarrow\psi, \\
\]
where $r\in \LAB_{arg}$ is a unique label, $A(r)=\set{\seq\phi}\subseteq
\LIT_{arg}$ is the antecedent of $r$, $C(r)=\psi\in \LIT_{arg}$ is the
consequent of $r$, and $\hookrightarrow\in \set{\TO, \To, \defeater}$ is the
type of $r$. \end{description} We use $R[q]$ to indicate all rules with
consequent $q$. We denote the sets of strict, rules, strict and defeasible
rules, and defeaters with $R_{\str}$, $R_{\de}$, $R_{\sd}$, and $R_{\dft}$,
respectively.
\end{defn}

\begin{defn}[Defeasible Argumentation Theory]
\label{def:strategic-theory}
A defeasible argumentation theory is a structure
\[
D_{arg}=(F, R,>)
\]
where
\begin{itemize}
\item $F \subseteq \LIT$ is a finite set of facts;
\item $R$  is the finite set of rules;
\item The superiority relation $>$ is acyclic, irreflexive, and asymmetric.
\end{itemize}
\end{defn}

\begin{defn}[Proofs]\label{def:Proofs}
  Given an agent theory $D$, a proof $P$ of length $n$ in $D$ is a finite  sequence $P(1),\ldots,P(n)$ of labelled formulas of the type $+\Delta q$,
  $-\Delta q$, $+\partial q$ and $-\partial q$, where the proof
  conditions defined in the rest of this section hold. $P(1..n)$ denotes the initial part of the derivation of length $n$.
\end{defn}

We start with some terminology. 
%
\begin{defn}\label{def:provableArg}
Given $\#\in\{\Delta,\partial\}$ and a proof $P$ in $D$, a literal $q$ is $\#$-\emph{provable} in $D$ if there is a line $P(m)$ of $P$ such that  $P(m)=+\# q$. A literal $q$ is $\#$-\emph{rejected} in $D$
if there is a line $P(m)$ of $P$ such that $P(m)=-\# q$.
\end{defn}
%
The definition of $\Delta$ describes just forward chaining of
strict rules:
\begin{center}
\begin{minipage}{.5\textwidth}%
\begin{tabbing}
$+\Delta$: \= If $P(n+1)=+\Delta q$ then  \+\\
  (1) $q\in \FACTS$ or\\
  (2) $\exists r\in R_{\str}[q]$ s.t. $\forall a \in A(r).~a$ is
      $\Delta$-provable.\-\\
$-\Delta$: \= If $P(n+1)=-\Delta q$ then \+\\
 (1) $q\notin \FACTS$ and\\
 (2)  $\forall r\in R_{\str}[q].~\exists a \in A(r)$ s.t. $a$ is
      $\Delta$-rejected.
\end{tabbing}
\end{minipage}
\end{center}
For a literal $q$ to be definitely provable either is a fact, or there is a
strict rule with head $q$, whose antecedents have all been definitely proved
previously. And to establish that $q$ cannot be definitely proven we must
establish that every strict rule with head $q$ has at least one antecedent is
definitely rejected.

The following definition is needed to introduce the defeasible provability.

\begin{defn}
A rule $r\in R_{sd}$ is \emph{applicable} in the proof condition for
$\pm\partial$ iff $\forall a\in A(r)$, $+\partial a\in P(1..n)$. A rule $r$ is
\emph{discarded} in the condition for $\pm\partial$ iff $\exists a\in A(r)$
such that $-\partial a\in P(1..n)$.	
\end{defn}

\begin{center}
\begin{minipage}{.3\textwidth}
\begin{tabbing}
  $+\partial$: If $P(n+1)=+\partial q$ then\\
  (1)\= $+\Delta q\in P(1..n)$ or \\
  (2) \= (2.1) \= $-\Delta\non q\in P(1..n)$ and\\
  \> (2.2) $\exists r\in R_{\sd}[q]$ s.t. $r$ is applicable,
        and \\
  \> (2.3) $\forall s\in R[\non q].$ either
        $s$ is discarded, or \\
      \>\> (2.3.1) \= $\exists t\in R[q]$ s.t. $t$ is applicable and
      $t>s$.\\
%
  $-\partial$: If $P(n+1)=-\partial q$ then\+\\
  (1) \= $-\Delta_X q\in P(1..n)$ and either\+\\
      (2.1) $+\Delta\non q\in P(1..n)$ or\\
      (2.2) \= $\forall r\in R_{\sd}[q].$ either $r$ is discarded, or\\
      (2.3) $\exists s\in R[\non q]$ s.t. $s$ is applicable, and\+\\
            (2.3.1) \= $\forall t\in R[q].$ either $t$ is discarded, or
                    $t\not> s$.
\end{tabbing}
\end{minipage}
\end{center}
%
%
%
To show that $q$ is defeasibly provable we have two choices: (1) We show that
$q$ is already definitely provable; or (2) we need to argue using the
defeasible part of a theory $D$. For this second case, $\non q$ is not
definitely provable (2.1), and there exists an applicable strict or defeasible
rule for $q$ (2.2). Every attack $s$ is either discarded (2.3), or defeated by
a stronger rule $t$ (2.3.1). $-\partial_X q$ is defined in an analogous manner
and follows the principle of \emph{strong negation} which is closely related to
the function that simplifies a formula by moving all negations to an inner most
position in the resulting formula, and replaces the positive tags with the
respective negative tags, and the other way around \cite{ecai2000-5}.

\subsection{Agent Logic}\label{subsec:agent_logic}

A defeasible agent theory is a standard defeasible theory enriched with 1) modes for rules, 2) modalities (belief, intention, obligation) for literals, and 3) relations for conversions
and conflict resolution. We report below only the distinctive features. For a detailed exposition see \cite{jaamas:bio}.

\begin{defn}[Language]
\label{def:language} Let $\PROP$  and $\LIT$ be a set of
propositional atoms and literals as in Definition~\ref{def:languageArg}, $\MOD=\{\BEL, \INT,\OBL \}$ be the set of modal
operators, and $\LAB_{\soc}$ be a set of labels. Define:
\begin{description}
\item[\emph{Modal literals}]
\[
\MODLIT= \set{Xl|l\in \LIT, X\in \set{\OBL, \INT}};
\]
\item[\emph{Rules}]
\[
r: \seq{\phi}\hookrightarrow_{X}\psi, \\
\]
where $r\in \LAB_{\soc}$ is a unique label, $A(r)=\set{\seq\phi}\subseteq \LIT
\cup \MODLIT$ is the antecedent of $r$, $C(r)=\psi\in \LIT$ is the consequent
of $r$, $\hookrightarrow\in \set{\TO, \To, \defeater}$ is the type of $r$, and
$X\in \MOD$ is the mode of $r$. \end{description} $R^{X}$ ($R^{X}[q]$) denotes
all rules of mode $X$ (with consequent $q$), and $R[q] = \bigcup_{X \in
\set{\BEL, \OBL, \INT}} R^{X}[q]$.
\end{defn}

\begin{observation}
	Rules for intention and obligation are meant to introduce modalities: for
example, if we have the intention rule $r: a \To_{\INT} b$ and we derive $a$,
then we obtain $\INT b$. On the contrary, belief rules produce literals and not
modal literals.
\end{observation}


\subsubsection{Rule conversion} It is sometimes meaningful to use rules for a
modality $Y$ as they were for another modality $X$, i.e., to convert one mode
of conclusions into a different one. Formally,
we define the asymmetric binary convert relation $\Convert \subseteq \MOD
\times \MOD$ such that $\Aconv{Y}{X}$ means `a rule of mode $Y$ can be used
also to produce conclusions of mode $X$'. This corresponds to the
following rewriting rule:
\[
\frac{\ds \seq{X a} \quad A(r)=\seq{a} \To_{Y} b}
{\ds X b} \mbox{\; $\Aconv{Y}{X}$} 
\]
where $A(r)\neq \emptyset$ and $A(r)\subseteq \LIT$.

\subsubsection{Conflict-detection/resolution}
We define an asymmetric binary conflict relation $\Conflict \subseteq \MOD \times \MOD$ such that $\Aconf{Y}{X}$ means `modes $Y$ and $X$ are in conflict and mode $Y$ prevails over $X$'.
\begin{defn}[Defeasible Agent Theory]
\label{def:agent-theory}
A defeasible agent theory is a structure
\[
D_{\soc}=(F_{\soc}, R^{\BEL}, R^{\INT}, R^{\OBL},>_{\soc}, \mathcal{V}, \mathcal{F})
\]
where
\begin{itemize}
\item $F_{\soc}\subseteq \LIT\cup \MODLIT$ is a finite set of facts;
\item $R^{\BEL}$, $R^{\OBL}$, $R^{\INT}$  are three finite sets of rules for beliefs, obligations, and intentions;
\item The superiority (acyclic) relation $>_{\soc}=>^{sm}_{\soc}\cup >^{\Conflict}_{\soc}$ such that: i. $>^{sm}_{\soc} \subseteq R^{X}\times R^{X}$ such that if $r>_{\soc}s$ then $r\in R^X[p]$ and $s\in R^{X}[\non p]$; and ii. $>^{\Conflict}_{\soc}$ is such that $\forall r\in R^Y[p], \forall s\in R^X[\non p]$ if $\Aconf{Y}{X}$ then $r>_{\soc}^{\Conflict}s$.

\item $\mathcal{V} = \set{\Aconv{\BEL}{\OBL},\Aconv{\BEL}{\INT}}$ is a set of convert relations;
\item $\mathcal{F} = \set{\Aconf{\BEL}{\OBL},\Aconf{\BEL}{\INT},\Aconf{\OBL}{\INT}}$ is a set of conflict relations.
\end{itemize}
\end{defn}

A proof is now a finite sequence of labelled formulas of the type $+\Delta_X q$,
  $-\Delta_X q$, $+\partial_X q$ and $-\partial_X q$.

The following
definition states the special status of belief rules, and that the
introduction of a modal operator corresponds to being able to derive
the associated literal using the rules for the modal operator.
\begin{defn}\label{def:provable}
Given $\#\in\{\Delta,\partial\}$ and a
  proof $P$ in $D$, $q$ is $\#$-\emph{provable} in $D$
  if there is a line $P(m)$ of $P$ such that either
\begin{enumerate}
\item $q$ is a literal and $P(m)=+\#_{\BEL}q$, or
\item $q$ is a modal literal $Xp$ and $P(m)=+\#_Xp$, or
\item $q$ is a modal literal $\neg Xp$ and $P(m)=-\#_Xp$.
\end{enumerate}
Instead, $q$ is $\#$-\emph{rejected} in $D$
if 
\begin{enumerate}\setcounter{enumi}{3}
\item $q$ is a literal and $P(m)=-\#_{\BEL}q$ or
\item $q$ is a modal literal $Xp$ and $P(m)=-\#_Xp$, or
\item $q$ is a modal literal $\neg Xp$ and $P(m)=+\#_Xp$.
\end{enumerate}
\end{defn}
%
We are now ready to report the definition of $\Delta_X$.
\begin{center}
\begin{minipage}{.5\textwidth}%
\begin{tabbing}
$+\Delta_X$: \= If $P(n+1)=+\Delta_X q$ then  \+\\
  (1) $q\in F$ if $X=\BEL$ or $Xq\in F$ or\\
  (2) $\exists r\in R^X_s[q]$ s.t. $\forall a \in A(r).~a$ is
      $\Delta$-provable or\\
  (3) \=$\exists r\in R^{Y}_s[q]$ s.t. $\Aconv{Y}{X}\in\mathcal{C}$ and\+\\
  $\forall a  \in A(r).~Xa$ is $\Delta$-provable.\-\-\\
$-\Delta_X$: \= If $P(n+1)=-\Delta_X q$ then \+\\
 (1) $q\notin F$ if $X=\BEL$ and $Xq\notin F$ and\\
 (2)  $\forall r\in R^{X}_s[q].~\exists a \in A(r)$ s.t. $a$ is
      $\Delta$-rejected and\\
 (3) \=$\forall r\in R^Y_s[q].$ if $\Aconv{Y}{X}\in\mathcal{C}$ then\+\\
  	$\exists a \in A(r)$ s.t. $Xa$ is $\Delta$-rejected.
\end{tabbing}
\end{minipage}
\end{center}
The sole difference with respect to $+\Delta$ is that now we may use rule of a
different mode, namely $Y$, to derive conclusions of mode $X$ through the
conversion mechanism. In this framework, only belief rules may convert to other
modes. That is the case, every antecedent of the belief rule $r\in R^{Y}$ in
clause (3) must be (definitely) proven with modality $X$.

We reformulate definition of being applicable/discarded, taking now into account also \Convert and \Conflict relations. 


\begin{defn}\label{def:applicable}

  Given a proof $P$ and $X,Y,Z\in \MOD$

\begin{itemize}

%
\item A rule $r$ is \emph{applicable} in the
  proof condition for $\pm\partial_{X}$ iff
\begin{enumerate}
\item  $r\in R^{X}$ and $\forall a\in A(r)$, $a$ is $\partial$-provable, or
\item $r\in R^Y$, $\Aconv{Y}{X}\in\mathcal{C}$, and $\forall a\in
  A(r)$, $Xa$ is $\partial$-provable.
\end{enumerate}
\item A rule $r$ is \emph{discarded} in the condition for
  $\pm\partial_{X}$ iff
  \begin{enumerate}\setcounter{enumi}{2}
  \item $r\in R^X$ and $\exists a\in A(r)$ such that $a$ is
    $\partial$-rejected; or
  \item $r\in R^Y$ and, if $\Aconv{Y}{X}$, then $\exists a\in A(r)$ such that $Xa$ is $\partial$-rejected, or
  \item $r\in R^Z$ and either $\neg \Aconv{Z}{X}$ or
    $\neg\Aconf{Z}{X}$.
  \end{enumerate}
\end{itemize}
\end{defn}

We are now ready to provide proof conditions for $\pm\partial_{X}$:

\begin{center}
%

\begin{minipage}{.3\textwidth}
\begin{tabbing}
  $+\partial_X$: If $P(n+1)=+\partial_X q$ then\\
  (1)\= $+\Delta_X q\in P(1..n)$ or \\
  (2) \= (2.1) \= $-\Delta_X\non q\in P(1..n)$ and\\
  \> (2.2) $\exists r\in R_{sd}[q]$ s.t. $r$ is applicable,
        and \\
  \> (2.3) $\forall s\in R[\non q]$ either
        $s$ is discarded, or \\
      \>\> (2.3.1) \= $\exists t\in R[q]$ s.t. $t$ is applicable and
      $t>s$, and\\
             \>\>\>either $t,s\in R^{Z}$, or $\Aconv{Y}{X}$ and $t\in R^{Y}$\\
 \\
%
  $-\partial_X$: If $P(n+1)=-\partial_X q$ then\+\\
  (1) \= $-\Delta_X q\in P(1..n)$ and either\+\\
      (2.1) $+\Delta_X\non q\in P(1..n)$ or\\
      (2.2) \= $\forall r\in R_{sd}[q]$, either $r$ is discarded, or\\
      (2.3) $\exists s\in R[\non q]$, s.t. $s$ is applicable, and\+\\
            (2.3.1) \= $\forall t\in R[q]$ either $t$ is discarded, or
                    $t\not> s$, or \\
             \> $t\in R^{Z},s\in R^{Z'}$, $Z\neq Z'$ and, \\
			 \>if $t\in R^{Y}$ then $\neg\Aconv{Y}{X}$.
\end{tabbing}
\end{minipage}
\end{center}
Again, the only difference with respect to $+\partial$ is that we have rules
for different modes, and thus we have to ensure the appropriate relationships
among the rules. Hence, clause (2.3.1) prescribes that either attack rule $s$
and counterattack rule $t$ have the same mode (i.e., $s,t\in R^{Z}$), or that
$t$ can be used to produce a conclusion of the mode $X$ (i.e., $t\in R^{Y}$ and
$\Aconv{Y}{X}$). Notice that this last case is reported for the sake of
completeness but it is useless in our framework since it plays a role only
within theories with more than three modes.

Being the strong negation of the positive counterpart, $-\partial_X q$ is
defined in an analogous manner.

We define the \emph{extension} of a defeasible theory as the set of all
positive and negative conclusions. In
\cite{DBLP:journals/tplp/Maher01,jaamas:bio}, authors proved that the extension
calculus of a theory in both argumentation and agent logic is linear in the
size of the theory.

Let us introduce some preliminary notions, which are needed for formulating the ``restoring sociality problem'' \cite{jaamas:bio} (and recalled below).

\begin{itemize}
  \item Given an agent defeasible theory $D$, a literal $l$ is \emph{supported} in $D$ iff
  there exists a rule $r\in R[l]$ such that $r$ is applicable,
  otherwise $l$ is not supported. For $X\in\MOD$ we use $+\Sigma_Xl$
  and $-\Sigma_{X}l$ to indicate that $l$ is supported / not supported
  by rules for $X$.
 \item \emph{Primitive intentions} of an agent are those intentions given as
facts in a theory.
 \item \emph{Primary} intentions and obligations are those derived using only rules for intentions and obligations (without any rule conversion).
\item A \emph{social agent} is an agent for which obligation rules are stronger than any conflicting intention rules but weaker than any conflicting belief rules. 
\end{itemize}

\subsection{Restoring Sociality Problem}\label{sec:restoring}
\begin{quote}
  \textsc{Instance}:\\
  Let $I$ be a finite set of primitive intentions, $\OBL p$ a primary
  obligation, and $D$ a theory such that $I\subseteq F$, $D\vdash
  -\partial_{\OBL}p$, $D\vdash -\Sigma_{\OBL}\non p$, $D\vdash
  +\partial_{\INT}\non p$, $D\vdash+\Sigma_{\OBL}p$ and
  $D\vdash-\Sigma_{\BEL}\non p$.

\textsc{Question:}\\
Is there a theory $D'$ equal to $D$ apart from containing only a
proper subset $I'$ of $I$ instead of $I$, such that $\forall q$ if
$D\vdash+\partial_{\OBL}q$ then $D'\vdash\partial_{\OBL}q$ and
$D'\vdash+\partial_{\OBL}p$?
\end{quote}

Let us the consider the theory consisting of
\begin{align*}
  F &=\set{\INT p, \INT s}\\
  R &=\set{r_1: p,s\To_{\BEL} q\quad r_2: {} \To_{\OBL} \non q\quad
           r_3: {}\To_{\BEL} s}\\
  > &=\set{r_1>r_2}
\end{align*}
$r_1$ is a belief rule and so the rule is stronger than the obligation
rule $r_2$. In addition we have that the belief rule is not applicable
(i.e., $-\Sigma_{\BEL}q$) since there is no way to prove
$+\partial_{\BEL} p$. There are no obligation rules for $q$, so
$-\partial_{\OBL}q$. However, rule $r_1$ behaves as an intention rule
since all its antecedent can be proved as intentions, i.e.,
$+\partial_{\INT}p$ and $+\partial_{\INT}s$. Hence, since $r_1$ is
stronger than $r_2$, the derivation of $+\partial_{\OBL}\non q$ is
prevented against the sociality of the agent.

The related decision problem is whether it is possible to avoid the
``deviant'' behaviour by giving up some primitive intentions,
retaining all the (primary) obligations, and maintaining a set of
primitive intentions 
as close as possible to the original set
of intentions.

\begin{thm}[\cite{jaamas:bio}]
 \label{th:NP-complete}
 The Restoring Sociality Problem is NP-complete.
\end{thm}


%% file: DialogueGames.tex
\section{Dialogue Games} 
\label{sec:dialogue_games}

The form of a \emph{dialogue game} involves a sequence of interactions between
two players, the \emph{Proponent} $\pr$ and the \emph{Opponent} $\op$. The
content of the dispute being that $\pr$ attempts to assess the validity
of a particular thesis (called \emph{critical literal} within our framework),
whereas $\op$ attacks $\pr$'s claims in order to refute such
thesis. We shift such position in our setting by stating that the Opponent has
the burden of proof on the opposite thesis, and not just the duty to refute the
Proponent's thesis.

The challenge between the parties is formalised by means of \emph{argument}
exchange. In the majority of concrete instances of argumentation frameworks,
arguments are defined as chains of reasoning based on facts and rules captured
in some formal language (in our case, a defeasible derivation $P$).
Each party adheres to a particular set of game rules as defined below.

The players partially shares knowledge of a defeasible theory. Each
participant has a private knowledge regarding some rules of the theory. Other
rules are known by both parties, but this set may be empty. These rules along
with all the facts of the theory and the superiority relation represent the
common knowledge of both participants.

By putting forward a private argument during a step of the game, the agent
increases the common knowledge by the rules used within the argument just
played.

Define the argument theory to be $D_{\Arg} = (\FACTS, R, >)$ such that i. $R=
R_{\pr} \cup R_{\op} \cup R_{\com}$, ii. $R_{\pr}$ ($R_{\op}$) is the private
knowledge of the Proponent (Opponent), and iii. $R_{\com}$ is the (possibly
empty) set of rules known by both participants. We use the superscript notation
$D_{\Arg}^{i}$, $R_{\pr}^{i}$, $R_{\op}^{i}$, and $R_{\com}^{i}$ to denote such
sets at turn $i$.

We assume that $D_{\Arg}$ is coherent and consistent, i.e., there is no literal
$p$ such that: i. $D_{\Arg}\vdash \pm \partial p$, and ii. $D_{\Arg}\vdash +
\partial p$ and $D_{\Arg}\vdash + \partial\non p$.

We now formalise the game rules, that is how the common theory $D_{\Arg}^{i}$
is modified based on the move played at turn $i$.

The parties start the game by choosing the critical literal $l$ to discuss
about: the Proponent has the burden to prove $+\partial l$ by using the current
common knowledge along with a subset of $R_{\pr}$, whereas the Opponent's final
goal is to prove $+\partial\non l$ using $R_{\op}$ instead of $R_{\pr}$.

The players may not present arguments in parallel: they take turn in making
their move.

The repertoire of moves at each turn just includes 1) putting forward an
argument, and 2) passing.

When putting forward an argument at turn $i$, the Proponent (Opponent) may
bring a demonstration $P$ whose terminal literal differs from $l$ ($\non l$).
When a player passes, she declares her defeat and the game
ends. This happens when there is no combination of the remaining private rules
which proves her thesis.

%

Hence, the initial state of the game is $T_{\Arg}^{0}=
(\FACTS,R^{0}_{\com},>)$ with $R^{0}_{\com}= R_{\com}$, and $R^{0}_{\pr}=
R_{\pr}$, $R^{0}_{\op}= R_{\op}$. 

If $T^{0}_{\Arg}\vdash + \partial l$, the Opponent starts the game. Otherwise,
the Proponent does so.

At turn $i$, if Proponent plays $R^{i}_{\Arg}$, then
\begin{itemize}
	\item $T^{i-1}_{\Arg}\vdash + \partial \non l$ ($T^{i-1}_{\Arg}\vdash - \partial l$ if $i=1$);
	\item $R^{i}_{\Arg}\subseteq R^{i-1}_{\pr}$;
	\item $T^{i}_{\Arg} = (F,R^{i}_{\com},>)$;
	\item $R^{i}_{\pr} = R^{i-1}_{\pr}\setminus R^{i}_{\Arg}$, $R^{i}_{\op} = R^{i-1}_{\op}$, and $R^{i}_{\com} = R^{i-1}_{\com}\cup R^{i}_{\Arg}$;
	\item $T^{i}_{\Arg}\vdash + \partial l$.
\end{itemize}

At turn $i$, if Opponent plays $R^{i}_{\Arg}$, then
\begin{itemize}
	\item $T^{i-1}_{\Arg}\vdash + \partial l$;
	\item $R^{i}_{\Arg}\subseteq R^{i-1}_{\op}$;
	\item $T^{i}_{\Arg} = (F,R^{i}_{\com},>)$;
	\item $R^{i}_{\pr} = R^{i-1}_{\pr}$, $R^{i}_{\op} = R^{i-1}_{\op}\setminus R^{i}_{\Arg}$, and $R^{i}_{\com} = R^{i-1}_{\com}\cup R^{i}_{\Arg}$;
	\item $T^{i}_{\Arg}\vdash + \partial\non l$.
\end{itemize}

\subsection{Strategic Argumentation Problem} 
\label{subsec:strategic_argumentation_problem}

\textsc{Proponent's instance for turn $i$:} Let $l$ be the critical literal, $R^{i-1}_{\pr}$ be the set of the private rules of the Proponent, and
$T^{i-1}_{\Arg}$ be such that either $T^{i-1}_{\Arg} \vdash -\partial l$ if $i=1$, or $D^{i-1}_{\Arg} \vdash +\partial \non l$ otherwise.

\noindent \textsc{Question:} Is there a subset $R^{i}_{\Arg}$ of
$R^{i-1}_{\pr}$ such that $D^{i}_{\Arg} \vdash +\partial l$?

\smallskip

\noindent\textsc{Opponent's instance for turn $i$:} Let $l$ be the critical literal, $R^{i-1}_{\op}$ be the set of the private rules of the Opponent, and
$D^{i-1}_{\Arg}$ be such that $D^{i-1}_{\Arg} \vdash +\partial l$.

\noindent \textsc{Question: } Is there a subset $R^{i}_{\Arg}$ of
$R^{i-1}_{\op}$ such that $D^{i}_{\Arg} \vdash +\partial \non l$?



%% file: Reduction.tex
\section{Reduction} 
\label{sec:reduction}


%
%
%
%

%

We now show how to transform Agent Logic (Section \ref{subsec:agent_logic}) into Argumentation Logic (Section \ref{subsec:argumentation_logic}). Basically, we need to act by transforming both literals and rules: whereas the agent theory deals with three different modes of rules and modal literals, the argumentation theory has rules without modes and literals.

The two main ideas of transformations proposed in
Definitions~\ref{def:langTransf} and \ref{def:theoryTransf} are
\begin{itemize}
	\item Flatten all modal literals with respect to internal negations
modalities. For instance, $\non p$ is flattened into the literal $not\_p$,
while $\OBL q$ is $obl\_q$.
	\item Remove modes from rules for $\BEL$, $\OBL$ and $\INT$. Thus, a rule
with mode $X$ and consequent $p$ is transformed into a standard, non-modal rule
with conclusion $Xp$. An exception is when we deal with belief rules, given
that they do not produce modal literals. Therefore, rule $\To_{\OBL} p$ is
translated in $\To obl\_p$, while rule $\To_{\BEL} q$ becomes $\To q$.
\end{itemize}

Function \pflat flattens the propositional part of a literal and syntactically
represents negations; function \fflat flattens modalities.

\begin{defn}\label{def:langTransf}
Let $D_{\soc}$ be a defeasible agent theory. Define two syntactic transformations $\pflat: \LIT_{\soc} \TO \PROP_{arg}$ and $\fflat: \MODLIT_{\soc}\cup \LIT_{\soc} \TO \LIT_{arg}$ as
\[
\Pflat p =
\begin{cases}
p\in \PROP_{arg} &\text{ if } p\in \PROP_{\soc}\\
not\_q\in \PROP_{arg} &\text{ if } p = \neg q,\, q \in \PROP_{\soc}	
\end{cases}
\]
\[
\Flatulo p = 
\begin{cases}
\Pflat q &\text{ if } p=q, \\
obl\_\Pflat q &\text{ if } p= \OBL q\\
\neg obl\_\Pflat q  &\text{ if } p= \neg\OBL q\\
int\_\Pflat q  &\text{ if } p= \INT q\\
\neg int\_\Pflat q  &\text{ if } p=\neg \INT q.\\
\end{cases}
\]
\end{defn}
Given that in BIO a belief modal literal is not $\BEL p$ but simply $p$, we
have that $\Flatulo p = \Pflat p$ whenever the considered mode is $\BEL$, while
$\Flat Xp = \flatx{x}{p}$ if \mbox{$X=\set{\OBL, \INT}$}.

We need to redefine the concept of complement to map BIO modal literals into an
argumentation logic with literals obtained through \fflat. Thus, if $q\in
\PROP_{arg}$ is a literal $p$ then $\non q$ is $not\_p$; and if $q$ is
$not\_p$, then $\non q$ is $p$. Moreover, if $q\in\LIT_{arg}$ is $\flatx{x}{p}$
then $\non q = \flatx{x}{\non p}$; and $q$ is $\neg \flatx{x}{p}$ then $\non q
= \flatx{x}{p}$.

We now propose a detailed description of facts and rules introduced by
Definition~\ref{def:theoryTransf}.

In the ``restoring sociality problem'' we have to select a subset of factual
intentions, while in the ``strategic argumentation problem'' we choose a subset
of rules to play to defeat the opponent's argument. Therefore, factual
intentions are modelled as strict rules with empty antecedent ($r_{p}$), while
factual beliefs and obligations are facts of $D_{arg}$.

We recall that, while proving $\pm\#_{X}q$, a rule in BIO may fire if either is
of mode $X$, through $\Convert$, or through $\Conflict$. Hence, a rule $r$ in
$D_{\soc}$ has many counterparts in $D_{arg}$.

Specifically, $r_{fl}$ is built from $r$ by: removing the mode, and flattening
each antecedent of $r$ as well as the consequent $p$ which in turn embeds the
mode introduced by $r$.

Moreover, if $r\in R^{\BEL}[p]$ then it may be used through conversion to
derive $Xp$. To capture this feature we introduce a rule $r_{Cvx}$ with
conclusion $\flatx{x}{p}$ and where for each antecedent $a\in A(r)$ the
corresponding in $A(r_{Cvx})$ is $\flatx{x}{a}$ according either to clause (3)
of $+\Delta_{X}$ or to condition 2. of Definition~\ref{def:applicable}.

In $D_{\soc}$, it is easy to determine which rule may fire against one another,
being that consequents of rules are non-modal literals. Even when the rules
have different modes and the conflict mechanism is used, their conclusions are
two complementary literals. Given the definition of complementary literals
obtained through \fflat we have introduced after
Definition~\ref{def:langTransf}, this is not the case for the literals in
$D_{arg}$. The situation is depicted in the following theory.

\begin{center}
	
\begin{tabular}{l@{\hskip 1cm}l}
$r: a \To_{\OBL}p$ & $r_{fl}: a \To obl\_p$\\
$s: b \To_{\INT}\neg p$ & $s_{fl}: b \To int\_not\_p$\\
$t: c \To_{\BEL} p$ & $t_{fl}: c \To p$.
\end{tabular}
\end{center}
Here, $r$ may fire against $s$ through $\Aconf{\OBL}{\INT}$ while $r_{fl}$
cannot, given that $obl\_p$ is not the complement of $int\_not\_p$. In the same
fashion, if we derive $+\partial_{\BEL} c$ then $t$ may fire against $s$
because of $\Aconf{\BEL}{\INT}$, while if we have either $+\partial_{\OBL} c$
or $+\partial_{\INT} c$ then the conflict between beliefs and intentions is
activated by the use of $r$ through either $\Aconv{\BEL}{\OBL}$ or
$\Aconv{\BEL}{\INT}$, respectively. Nonetheless, in both cases there is no
counterpart of $t$ in $D_{arg}$ able to fire against $int\_not\_p$.

To obviate this issue, we introduce a defeater $r_{CfOI}$ where we flatten the
antecedents of $r$ and the conclusion is the intention of the conclusion of
$r$, namely $\flatx{int}{C(r)}$. This means that when $r$ fires, so does
$r_{CfOI}$ attacking $s_{fl}$. Notice that being $r_{CfOI}$ a defeater, such a
rule cannot derive directly $+\partial \flatx{int}{p}$ but just prevents the
opposite conclusion. The same idea is adopted for rules $r_{Cfbelx}$ and
$r_{CvyCfx}$: defeaters $r_{Cfbelx}$ are needed to model conflict between
beliefs and intentions (as rule $t$ in the previous example), whereas defeaters
$r_{CvyCfx}$ take care of situations where $r\in R^{Z}$ may be used to convert
$Z$ into $Y$ and $Z$ prevails over $X$ by \Conflict.

Thus in the previous example, we would have: $r_{CfOI}: a \defeater int\_p$,
$t_{Cfbelint}: c \defeater int\_p$, $t_{Cfbelint}: c \defeater int\_p$,
$t_{CvxCfint}: x\_c \defeater int\_p$, with $x\in \set{obl, int}$.

Antecedents in BIO may be negation of modal literals; in that framework, a
theory proves $\neg Xp$ if such theory rejects $Xp$ (as stated by condition 3.
of Definition~\ref{def:provable}). In $D_{arg}$ we have to prove $\neg
\flatx{x}{p}$ This is mapped in $D_{arg}$ through conditions
\ref{def2.dum-xp}--\ref{def2.neg-xp} of Definition~\ref{def:theoryTransf} and
the last condition of $>$.

\begin{defn}\label{def:theoryTransf}
	Let \mbox{$D_{\soc} =(\FACTS_{\soc}, R^{\BEL}, R^{\OBL}, R^{\INT}, >_{\soc},
	\mathcal{V}, \mathcal{F})$} be a defeasible agent theory. Define \mbox{$D_{arg} = (\FACTS,R,>)$} an argumentation defeasible theory such that 

\begin{align}
\FACTS &= \set{\Flatulo p| p \in \FACTS_{\soc}, p \in \LIT \text{ or } p=\OBL q}\label{def2.F}\\
R &= \set{r_{p}: { } \TO int\_\Pflat p| \INT p \in \FACTS_{\soc}}\label{def2.p}\\
&\cup \set{r_{fl}: \bigcup_{a\in A(r)} \Flatulo a \hookrightarrow \Flatulo p| r\in R^{X}[q], \nonumber\\ 
&\quad X=\BEL  \text{ and } p=q, \text{ or } p=Xq \in \MODLIT}\label{def2.fl}\\
&\cup \set{r_{Cvx}: \bigcup_{a\in A(r)} x\_\Pflat{a} \hookrightarrow \flatx{x}{p} | r\in R_{\sd}^{\BEL}[p], \nonumber\\
&\quad A(r) \neq \emptyset, A(r)\subseteq \LIT, x \in \set{obl, int}}\label{def2.Conv}\\
&\cup \set{r_{CvyCfx}: \bigcup_{\flatx{y}{a}\in A(r_{Cvy})} \flatx{y}{a} \defeater \flatx{x}{p} | \nonumber\\
&\quad r_{Cvy}\in R[\flatx{y}{p}], x, y\in \set{obl, int}, x\neq y}\label{def2.CvxCfy}\\
&\cup \set{r_{Cfbelx}: \bigcup_{a \in A(r)}\Flatulo{a} \defeater \flatx{x}{p} | r\in R^{\BEL}[p],\nonumber\\ 
&\quad x\in \set{obl, int}}\label{def2.ConflBelX}\\
&\cup \set{r_{CfOI}:  \bigcup_{a \in A(r)}\Flatulo{a} \defeater \flatx{int}{p} | r\in R^{\OBL}[p]}\label{def2.ConflOI}\\
&\cup \set{r_{dum-xp}: \flatx{x}{p} \To xp  | r\in R^{Y}. \neg Xp\in A(r)}\label{def2.dum-xp}\\
&\cup \set{r_{dum-negxp}: { } \To \non xp | r_{dum-xp}\in R}\label{def2.dum-negxp}\\
&\cup \set{r_{neg-xp}: \non xp \To \neg \flatx{x}{p} | r_{dum-negxp}\in R}\label{def2.neg-xp}\\
> &= \set{(r_{\alpha},s_{\beta})| (r,s) \in >_{\soc}, \alpha,\beta\in\{fl, Cvx, CvxCfy,\nonumber\\
&\quad Cfbelx,CfOI}\}\nonumber\\
&\cup \set{(r_{fl},s_{neg-xp})| r_{fl}\in R[\flatx{x}{p}]}\nonumber\\
&\cup \set{(r_{dum-xp},s_{dum-negxp})| r_{dum-xp},s_{dum-negxp} \in R}.\label{def2.sup}
\end{align}
We name $D_{\Arg}$ the \emph{argumentation counterpart} of $D_{\soc}$.
\end{defn}

The following result is meant to prove the correctness of the
transformation given in Definition~\ref{def:theoryTransf}. This is the case
when the transformation preserves the positive and negative provability for any
given literal.

\input{Proof.tex}

In order to show the final result that the Strategic Argumentation Problem is
NP-Complete, we first prove that the proposed transformation is polynomial.

\begin{thm}\label{thm:linearTrans}
	There is a linear transformation from any defeasible agent theory $T_{\soc}$
to its argumentation counterpart $T_{\Arg}$.
\end{thm}
\begin{proof}
	The transformation rules of Definition~\ref{def:theoryTransf} are applied
once to each rule and each tuple of the superiority relation. Transformation
rule (1) maps one fact in $T_{\soc}$ into one fact in $T_{\Arg}$. Transformation
rule (2) maps one primitive intention $T_{\soc}$ into one strict rule in
$T_{\Arg}$. Rule (3) and (7) again copy one rule into one rule. Rules (4)--(6)
generate two rules in $T_{\Arg}$ for every belief rule in $T_{\soc}$. Rules
(8)--(10) generate a total of three rules in $T_{\Arg}$ for each negative modal
literal in $T_{\soc}$. Rule (11) generates thirty-two tuples in $T_{\Arg}$ for
each tuple in $>_{\soc}$ and two tuples for each negative modal literal in in
$T_{\soc}$.

The above reasoning shows that the transformation performs a number of steps
that is, in the worst case, smaller than thirty-two times the size of the
defeasible agent theory, and this proves the claim.
\end{proof}

\begin{thm}\label{thm:NP-complete}
	The Strategic Argumentation Problem is NP-Complete.
\end{thm}
\begin{proof}
	First, the Strategic Argumentation Problem is polynomially solvable on
non-deterministic machines since, given a defeasible argumentation theory
$D_{\Arg}$, we guess a set of rules $R_{\Arg}^{i}$ and we can check the
extension in polynomial time \cite{DBLP:journals/tplp/Maher01}.
	
	Second, the Strategic Argumentation Problem is NP-hard. In fact, we map the
Restoring Sociality Problem \cite{jaamas:bio} into the Strategic Argumentation
Problem. Given a (deviant) defeasible agent theory $D_{\soc}$, $D_{\soc}$ is
mapped into its argumentation counterpart $D_{\Arg}$
(Definition~\ref{def:theoryTransf}). The transformation is polynomial
(Theorem~\ref{thm:linearTrans}) and correct (Theorem~\ref{thm:transSound}).
\end{proof}


\subsection{Discussion}
\label{ssec:discussion}
In this paper we concentrated in a game with a symmetry on what the two parties
have to prove: $\pr$ has to prove $l$ (i.e., $+\partial l$) while $\op$ has to
prove $\non l$ (i.e., $+\partial \non l$); however, it is possible to have
games where the two parties have different burden on proof, namely, the
proponent $\pr$ has to prove $l$ and the opponent $\op$ has to disprove it. In
Defeasible Logic this can be achieve either by proving that the opposite holds,
namely $+\partial \non l$ or simply by showing that $l$ is not provable, i.e.,
$-\partial l$. In this case we have two different types of strategic
argumentation problems: one for the proponent (which is the same as the current
one), and one for the opponent. For the opponent, the related decision problem
is if there exists a subset of her private rules such that adding it to current
public rule make that the resulting theory proves $-\partial l$. The proof
conditions for $+\partial$ and $-\partial$ are the strong negation of each
other \cite{ecai2000-5}; hence this version of the strategic argumentation
problem is coNP-complete.

The NP-completeness result of the paper is proved for the ambiguity blocking,
team defeat variant of Defeasible Logic. However, the proof of the result does
not depend on the specific features of this particular variant of
the logic, and the result extends to the other variants of the logic (see
\cite{inclusion} for the definition of the various variants). The version of
the argumentation logic presented in this paper does not correspond to the
grounded semantics for Dung's style abstract argumentation framework (though it
is possible to give such a semantics for it, see \cite{GovernatoriMAB04}).
However, the ambiguity blocking variant corresponds to Dung's grounded
semantics \cite{GovernatoriMAB04}. Accordingly, strategic argumentation seems
to be a computationally infeasible problem in general.

Finally, in our game we chose that the superiority relation is known \emph{a
priori} by both players. If not so, the problem reduces to revising the
corresponding Agent Logic by changing a combination of rules and superiority
relation. The problem of revising a defeasible theory by only modifying the
superiority relation has proven to be NP-complete in \cite{Pref}.

\section{Summary} 
\label{sec:related_works_&_conclusions}

Almost all research in AI on argumentation assumes that strategic dialogues are
games of complete information, i.e., dialogues where the structure of the game
is common knowledge among the players. Following \cite{OkunoT09,Sato:2011} we
argue that argument games work under incomplete information: each player does
not know the other player's knowledge, thus she cannot predict which arguments
are attacked and which counterarguments are employed for attacking the
arguments; hence, argument moves can disclose such private information, thus
allowing the other player to attack.

While it is outside the scope of this paper how to analyse strategic dialogues
in game-theoretic terms, our research effort is preliminary to this analysis,
since it studies the computation cost for logically characterising the problem
that any argumentation game with incomplete information potentially rises. We
have shown that the problem of deciding what set of rules to play (``Strategic
Argumentation Problem'') at a given move is NP-complete even when the problem
of deciding whether a given theory (defeasibly) entails a literal can be
computed in polynomial time. To this end, we mapped the NP-complete ``Restoring
Sociality Problem'' proposed in \cite{jaamas:bio} into the strategic
argumentation problem.






%% file: Proof.tex
\begin{thm}\label{thm:transSound}
	Let $D_{\soc} =(\FACTS_{\soc}, R^{\BEL}, R^{\OBL}, R^{\INT}, >_{\soc},
\mathcal{V}, \mathcal{F})$ be a defeasible agent theory and $D_{\Arg} =
(\FACTS,R,>)$ the argumentation counterpart of $D_{\soc}$. Given $p\in \LIT \cup \MODLIT$ and $\# = \set{\Delta, \partial}$:
\begin{enumerate}
	\item $D_{\soc}\vdash \pm \#_{\BEL} p$ iff $D_{\Arg}\vdash \pm \# \Flatulo p$;
	\item $D_{\soc}\vdash \pm \#_{X} p$ iff $D_{\Arg}\vdash \pm \# \Flat Xp$, $X\in \set{\OBL, \INT}$.
\end{enumerate}	
\end{thm}

\begin{proof}\label{pf:transSound}
	The proof is by induction on the length of a derivation $P$. For the
inductive base, we consider all possible derivations of length 1 for a given
literal $q$. Given the proof tags' specifications as in Definitions \ref{def:langTransf} and \ref{def:theoryTransf},
the inductive base only takes into consideration derivations for $\pm \Delta$,
since to prove $\pm \partial q$ requires at least 2 steps.
	
\paragraph{$P(1)= +\Delta_{X} q$.} This is possible either when
clause (1), or (2) of $+\Delta_{X}$ in $D_{\soc}$ holds. 

For (1), we have either i. $q \in \FACTS_{\soc}$ and $X=\BEL$ or $ \OBL q \in
\FACTS_{\soc}$ then $\Flatulo q \in \FACTS$ or $\Flat \OBL q \in \FACTS$ by
condition (\ref{def2.F}) of Definition~\ref{def:theoryTransf}; or ii. $\INT q
\in \FACTS_{\soc}$ then there exists $r_{q}\in R_{s}[int\_\Pflat q]$,
$A(r_{q})=\emptyset$, by condition (\ref{def2.p}) of
Definition~\ref{def:theoryTransf}. Cases i. and ii. as seen together state that
either if $X = \BEL$ then $D_{arg}\vdash +\Delta \Flat q$, or $D_{arg}\vdash
+\Delta \Flat X q$ otherwise, by clause (1), or (2) of $+\Delta$ in $D_{arg}$.

Concerning (2) of $+\Delta_{X}$, there exists $r\in R^{X}_{s}[q]$ such that
$A(r)=\emptyset$. Hence, if $X=\BEL$ then we have $r_{fl}\in R_{s}[\Flatulo
q]$, otherwise we have $r_{fl}\in R_{s}[\flatx{x}{q}]$ with $x=\set{obl, int}$,
where both situations follow by condition (\ref{def2.fl}) of
Definition~\ref{def:theoryTransf} and $A(r_{fl})=\emptyset$. Thus,
$D_{arg}\vdash +\Delta \Flat q$ or $D_{arg}\vdash +\Delta \Flat X q$,
respectively, by clause (2) of $+\Delta$ in $D_{arg}$.

\paragraph{$P(1)= +\Delta \Flatulo q$.} This is possible either
when clause (1), or (2) of $+\Delta$ in $D_{arg}$ holds. 

For (1), we have either $\Pflat q \in \FACTS$ with $q=p$ and $p\in \LIT$, or
$obl\_\Pflat p \in \FACTS$ with $q= \OBL p$; hence, by
Definition~\ref{def:langTransf} and condition (\ref{def2.F}) of
Definition~\ref{def:theoryTransf}, we conclude that $p\in \FACTS_{\soc}$ or
$\OBL p \in \FACTS_{\soc}$, respectively. Thus, either $D_{\soc}\vdash
+\Delta_{\BEL} p$ or $D_{\soc}\vdash +\Delta_{\OBL} p$ by clause (1) of
$+\Delta_{X}$ in $D_{\soc}$.

Concerning (2) of $+\Delta_{X}$, we consider if either i. $q=p$ and $p\in \LIT$
or $q = \OBL p$, or ii. $q = \INT p$.

Case i., there exists $r_{fl}\in R_{s}[\Flatulo p]$,
$A(r_{fl})=\emptyset$. Therefore, there exists $r\in R^{X}_{s}[p]$, with
$A(r)=\emptyset$ and $X = \set{\BEL, \OBL}$, by condition (\ref{def2.fl}) of
Definition~\ref{def:theoryTransf}. Thus, $D_{\soc}\vdash +\Delta_{X} p$ by
clause (2) of $+\Delta_{X}$ in $D_{\soc}$. 

Case ii., two possible situations arise: a) There exists $r_{fl}\in
R_{s}[int\_\Pflat p]$, $A(r_{fl})=\emptyset$, then there exists $r\in
R^{\INT}_{s}[p]$, $A(r)=\emptyset$, by condition (\ref{def2.fl}) of
Definition~\ref{def:theoryTransf}; or b) there exists $r_{p} \in
R_{s}[int\_\Pflat p]$, then $\INT p \in \FACTS_{\soc}$ by condition
(\ref{def2.p}) of Definition~\ref{def:theoryTransf}. For a) as well as for b),
$D_{\soc}\vdash +\Delta_{\INT} p$ by clause (2) or (1), respectively, of
$+\Delta_{X}$ in $D_{\soc}$.

\paragraph{$P(1)= -\Delta_{X} q$, $P(1)= -\Delta \Flatulo q$.} Both
demonstrations are the same as and use the same ideas of cases $P(1)=
+\Delta_{X} q$ or $P(1)= +\Delta \Flatulo q$, respectively.

($P(1)= -\Delta_{X} q$) Clause (1) and (2) of $-\Delta_{X}$ in $D_{\soc}$ are
satisfied. Thus, $q \not\in \FACTS_{\soc}$ ($q\in \LIT$) or $Xq \not\in
\FACTS_{\soc}$ and, consequently, neither $\Flatulo q \in \FACTS$ nor
$obl\_\Pflat q\in \FACTS$, and there is no rule $r_{q}$ that proves $\INT q$.
Moreover, for all $r\in R^{X}_{s}[q]$ then $A(r)\neq \emptyset$ and,
accordingly, the same situation holds for all the corresponding rules of type
$r_{fl}$ in $D_{arg}$.

The same reasoning applies for the other direction. 

\paragraph{$P(n+1)= +\Delta_{X} q$.} If $q\in \FACTS_{\soc}$ and $X=\BEL$,
or $Xq \in \FACTS_{\soc}$ and $X=\set{\OBL, \INT}$, then the case is the same as the corresponding inductive base.

If there exists $r\in R^{X}_{s}[q]$ such that $a$ is $\Delta$-provable at
$P(n)$, for all $a\in A(r)$, meaning that: a) There exists $r_{fl}: \bigcup_{a\in A(r)} \Flatulo a \TO \Flatulo q$
with $X=\BEL$, or there exists $r_{fl}: \bigcup_{a\in A(r)} \Flatulo a \TO
\flatx{x}{q}$ with $x \in \set{obl, int}$ by condition (\ref{def2.fl}) of
Definition~\ref{def:theoryTransf}; and b) $\Flatulo a$ is
$\Delta$-provable for all $\Flatulo a \in A(r_{fl})$, by inductive hypothesis.
Hence, $D_{arg}\vdash +\Delta \Flatulo q$ or $D_{arg}\vdash +\Delta \Flat Xq$,
respectively, by clause (2) of $+\Delta$ in $D_{arg}$.

Finally, if clause (3) of $+\Delta_{X}$ in $D_{\soc}$ is the case, then there
exists $r\in R^{Y}_{s}[q]$ such that $\Aconv{Y}{X} \in \mathcal{V}$ and $Xa$ is
$\Delta$-provable at $P(n)$, for all $a\in A(r)$. Thus, there exists $r_{Cvx}\in
R_{s}[\flatx{x}{q}]$ by condition (\ref{def2.Conv}) of
Definition~\ref{def:theoryTransf}, and $\flatx{x}{a}$ is $\Delta$-provable for
all $\flatx{x}{a} \in A(r_{Cvx})$, by inductive hypothesis. Again,
$D_{arg}\vdash +\Delta \Flat Xq$ by clause (2) of $+\Delta$ in $D_{arg}$.

\paragraph{$P(n+1)= +\Delta \Flatulo q$.} This is possible either
when clause (1), or (2) of $+\Delta$ in $D_{arg}$ holds.

If $\Flatulo q\in \FACTS$, then the proof is the same as the corresponding
inductive base.

Otherwise, we consider if either i. $q=p$ and $p\in \LIT$, or ii. $q = Xp$ with
$X=\set{\OBL, \INT}$. 

Case i., there exists $r_{fl}\in R_{s}[\Flatulo p]$, such that $\Flatulo a$ is
$\Delta$-provable at $P(n)$, for all $\Flatulo a \in A(r_{fl})$. Therefore: a)
There exists $r\in R^{\BEL}_{s}[p]$ by condition (\ref{def2.fl}) of
Definition~\ref{def:theoryTransf}; and b) $a$ is $\Delta$-provable for all
$a\in A(r)$ by inductive hypothesis. Thus, $D_{\soc}\vdash +\Delta_{\BEL} p$ by
clause (2) of $+\Delta_{X}$ in $D_{\soc}$.

Case ii. is divided in two sub-cases. First sub-case, there exists $r_{fl}\in
R_{s}[\flatx{x}{p}]$ such that $\Flatulo a$ is $\Delta$-provable at $P(n)$, for
all $\Flatulo a \in A(r_{fl})$. This case is analogous to the previous 
case. Second sub-case, there exists $r_{Cvx}\in R_{s}[\flatx{x}{p}]$, with
$x=\set{obl, int}$, such that $\flatx{x}{a}$ is $\Delta$-provable for all
$\flatx{x}{a} \in A(r_{Cvx})$. Therefore, the following two conditions are
satisfied: a) There exists $r\in R^{\BEL}_{s}[p]$ by condition
(\ref{def2.Conv}) of Definition~\ref{def:theoryTransf}, and b) $Xa$ is
$\Delta$-provable for all $a\in A(r)$ by inductive hypothesis. Thus,
$D_{\soc}\vdash +\Delta_{X} p$ by clause (3) of $+\Delta_{X}$ in $D_{\soc}$.

\paragraph{$P(n+1)= -\Delta_{X} q$.} Clauses (1)--(3) of $-\Delta_{X}$ in
$D_{\soc}$ hold. 

For (1), $q\not\in \FACTS_{\soc}$ ($q\in \LIT$) or $Xq \not\in
\FACTS_{\soc}$. Consequently, neither $\Flatulo q \in \FACTS$ nor
$obl\_\Pflat q\in \FACTS$, and there is no rule $r_{q}$ to support $int\_\Pflat q$. 

For (2), for all $r\in R^{X}_{s}[q]$ there exists $a\in A(r)$ such that $a$
is $\Delta$-rejected  at $P(n)$. Accordingly, for all the corresponding rules of type
$r_{fl}$ in $D_{arg}$, there exists $\Flatulo a\in A(r_{fl})$ which is
$\Delta$-rejected by inductive hypothesis. Hence, we conclude that
$D_{arg}\vdash -\Delta \Flatulo q$ if $X=\BEL$. 

Finally, the same reasoning applies for all rules $r\in R_{s}^{Y}[q]$, with
$\Aconv Y X \in\mathcal{V}$, where there exists $a\in A(r)$ such that $Xa$ is
$\Delta$-rejected at $P(n)$. Thus, we conclude that $D_{arg}\vdash -\Delta
\Flat Xq$, with $X=\set{\OBL, \INT}$.

\paragraph{$P(n+1)= -\Delta \Flat q$.} The proof follows the inductive base and
the case $P(n)= -\Delta_{X} q$.

\paragraph{$P(n+1)= +\partial_{X} q$.} Clauses (1) and (2.1) of
$+\partial_{X}$ have already been proved for the inductive step of
$\pm\Delta_{X}$. 

If clause (2.2) of $+\partial_{X}$ is the case, then there exists $r\in
R_{\sd}[q]$ such that $r$ is applicable at $P(n+1)$ (i.e., $a$ is
$\partial$-provable at $P(n)$ in $D_{\soc}$, for all $a\in A(r)$) and either
clause (2.3) or (2.3.1) is satisfied.

We have two cases. If $r\in R^{X}$ then there exists either $r_{fl}\in
R_{\sd}[\Flatulo q]$ when $X=\BEL$, or $r_{fl}\in R_{\sd}[\flatx{x}{q}]$
otherwise by condition (\ref{def2.fl}) of Definition~\ref{def:theoryTransf}.
Thus, $\Flatulo a$ is $\partial$-provable at $P(n)$ in $D_{arg}$, for all
$\Flatulo a\in A(r_{fl})$ by inductive hypothesis. If $r\in R^{Y}$ and
$X=\set{\OBL, \INT}$, then there exists either $r_{Cvx}\in R_{\sd}[x\_\Flatulo
q]$ by condition (\ref{def2.Conv}) of Definition~\ref{def:theoryTransf}. Hence,
$\flatx{x}{a}$ is $\partial$-provable at $P(n)$ in $D_{arg}$, for all $\Flatulo
a\in A(r_{Cvx})$ by inductive hypothesis. We conclude that clause (2.2) of
$+\partial$ holds in $D_{arg}$ by inductive hypothesis.

For, clause (2.3) if $s\in R[\non q]$ is discarded, then we have the following cases.

\begin{enumerate}[a.]
	\item $s\in R^{X}$, then there exists $a\in A(s)$ which is $\partial$-rejected at $P(n)$. Thus, $\Flatulo a$ is $\partial$-rejected at $P(n)$ in $D_{arg}$ by inductive hypothesis, and therefore $s_{fl}$ is discarded in $D_{arg}$.
	\item $s\in R^{\BEL}$ and $X\in \set{\OBL, \INT}$, then there exists $a\in A(s)$ such that $Xa$ is $\partial$-rejected at $P(n)$. Hence, $\flatx{x}{a}$ is $\partial$-rejected at $P(n)$ in $D_{arg}$ by inductive hypothesis and we conclude that $s_{Cvx}$ is discarded in $D_{arg}$.
	\item $X=\BEL$ and $s\in R^{Z}$ with $Z\in\set{\OBL,\INT}$, or $X=\OBL$ and
$s\in R^{\INT}$. We conclude that $s_{fl}$ is discarded because either $X=\BEL$
and $s_{fl}\not\in R[\non \Flatulo q]$, or $s_{fl}\not\in R[\non \flatx{x}{q}]$
otherwise.
\end{enumerate}

Finally, we consider clause (2.3.1) of $+\partial_{X}$. Following the above
reasoning, if $t$ is applicable in $D_{\soc}$, then $t_{fl}$ or $t_{Cvx}$ is
applicable in $D_{arg}$ as well.

If $t, s\in R^{Z}$ and $t >^{sm}_{\soc} s$, then $t_{\alpha} > s_{\alpha}$ with
$\alpha\in \set{fl, Cvz}$ by condition (\ref{def2.sup}) of
Definition~\ref{def:theoryTransf}. If $X =\BEL$, there is no need for further
analysis given that the transformation does not produce additional rules for
$\Pflat q$, for any literal $q$.

Otherwise, we have either 

\begin{enumerate}[i.]
	\item $t\in R^{\BEL}[q]$ and $s\in R^{X}[\non q]$: thus $t_{Cvx} > s_{fl}$;

	\item $t\in R^{\BEL}[q]$ and $s\in R^{X}[\non q]$: thus $t_{Cfbelx} > s_{fl}$, with $t_{Cfbelx}: \bigcup \Flatulo a \defeater \flatx{x}{q}$;

	\item $s, t\in R^{\BEL}$: thus either a) $t_{Cfbelx} > s_{Cvx}$, or b) $t_{CvyCfx} > s_{Cvx}$, with $t_{CvyCfx}: \bigcup \flatx{y}{a}\defeater \flatx{x}{q}$;
\end{enumerate} 
by condition (\ref{def2.sup}) and $\Aconf{\BEL}{X}$ for i. and ii.,
$t>^{sm}_{\soc} s$ for the last case. It only remains to prove that $t_{Cfbelx}$
and $t_{CvyCfx}$ are applicable in $D_{arg}$. If $t$ is applicable in $D_{\soc}$
at $P(n+1)$, then any $a\in A(t)$ is $\partial$-provable in $D_{\soc}$ at $P(n)$
and so is $\Flatulo a$ in $D_{arg}$ by inductive hypothesis. We conclude that
$t_{Cfbelx}$ is applicable in $D_{arg}$ at $P(n+1)$. Instead, if $t$ is
applicable in $D_{\soc}$ at $P(n+1)$ through $\Aconv{\BEL}{Y}$, then $Ya$ is
$\partial$-provable in $D_{\soc}$ at $P(n)$ for every $a\in A(t)$. By inductive
hypothesis, any $\flatx{y}{a}\in A(t_{CvyCfx})$ is $\partial$-provable as well.
Hence, $t_{CvyCfx}$ is applicable in $D_{arg}$ as well.

This completes the analysis when $s_{\alpha}$ with $\alpha\in \set{fl, Cvx}$;
we now analyse other possible attacks in $D_{arg}$ and first proceed for
$X=\OBL$, then for $X=\INT$.

Suppose there is a rule $w \in R^{\BEL}[\non q]$; $w$ produces rules
$w_{Cfbelx}$ and $w_{CvyCfx}$. In the first case $w$ would fire against $Xq$
due to $\Aconf{\BEL}{X}$. If $w$ is discarded in $D_{\soc}$ at $P(n+1)$, then
there exists $a\in A(w)$ such that $a$ is $\partial$-rejected in $D_{\soc}$ at
$P(n)$. By inductive hypothesis, we conclude that $\Flatulo a \in
A(w_{Cfbelx})$ is $\partial$-rejected in $D_{arg}$ at $P(n)$. Otherwise, $w$ is
defeated by an applicable $t$ in $D_{\soc}$. Assume
there is no $t\in R^{\BEL}[q]$ stronger than $w$. Thus, $D_{\soc} \vdash
-\partial_{X} q$, against the hypothesis. Therefore, $t>_{\soc} w$ and the
corresponding of $t$ in $D$ is stronger than $w_{Cfbelx}$ by construction of
$>$ in Definition~\ref{def:theoryTransf}.

An analogous reasoning applies for $w_{CvyCfx}$. Here, $w$ would be applicable
through $\Aconv{\BEL}{Y}$ and then fire against $Xq$ by $\Aconf{\BEL}{X}$. If
$w$ is discarded, then there exists $a\in A(w)$ such that $Ya$ is
$\partial$-rejected in $D_{\soc}$ at $P(n)$. By inductive hypothesis,
$\flatx{y}{a}$ is $\partial$-rejected in $D_{arg}$ and $w_{CvyCfx}$ is
discarded at $P(n+1)$. Otherwise, $w$ is defeated in
$D_{\soc}$ by an applicable $t\in R^{\BEL}[q]$ either directly, or through
conversion. In both cases, the corresponding rule of $t$ in $D_{arg}$ is
stronger than $w_{CvyCfx}$ by construction of $>$ in
Definition~\ref{def:theoryTransf}. Notice that if $w$ is applicable in
$D_{\soc}$ at $P(n+1)$ then $Ya$ is $\partial$-provable at $P(n)$ for any $a\in
A(w)$ and, consequently, so is $\flatx{y}{a}$ by inductive hypothesis, making
$w_{CvyCfx}$ applicable in $D_{arg}$.


A final analysis is in order when $X=\INT$ and we consider $w_{CfOI}: \bigcup
\Flatulo a \defeater \flatx{int}{\non q}$. The counterpart in $D_{\soc}$ is
$w\in R^{\OBL}[\non q]$, which is either discarded, or defeated by a stronger
rule $t$. Again, if $w$ is discarded in $D_{\soc}$, then there exists $a\in
A(w)$ such that $a$ is $\partial$-rejected in $D_{\soc}$ at $P(n)$. By inductive
hypothesis, $\Flatulo a$ is $\partial$-rejected in $D_{arg}$ and $w_{CfOI}$ is
discarded at $P(n+1)$. Otherwise $w$ is defeated either by an applicable $t\in
R^{\OBL}[q]$, or $t\in R^{\BEL}[q]$ (in this last case directly, or through
$\Aconf{\BEL}{\OBL}$, or through $\Aconf{\BEL}{\INT}$). This relation is
preserved in $>$ between $w_{CfOI}$ and the corresponding rule of $t$ in
$D_{arg}$ by condition (\ref{def2.sup}) of Definition~\ref{def:theoryTransf}.
If the $t$ is in $R^{\OBL}$, then $t_{CfOI}\in R[\flatx{int}{q}]$. Stating that
$t$ is applicable in $D_{\soc}$ at $P(n+1)$ means that every antecedent $a$ is
$\partial$-provable. By inductive hypothesis, so is the corresponding $\Flatulo
a$ in $D_{arg}$, making $t_{CfOI}$ applicable at $P(n+1)$.


\paragraph{$P(n+1)= +\partial \Flatulo q$.} Clauses (1) and (2.1) of
$+\partial$ have already been proved for the inductive base of
$\pm\Delta$. 

If $q=p$ and $p\in \LIT$, then the only rules to consider as support/attack
$\Flatulo p$ are obtained through condition (\ref{def2.fl}) of
Definition~\ref{def:theoryTransf}. Therefore, by
inductive hypothesis, for any applicable rule $r_{fl}$ the corresponding
rule $r$ in $D_{\soc}$ is applicable as well, and the same reasoning holds for
discarded rules. Moreover, the superiority relation is isomorphic for
such rules. Hence, $D_{\soc}\vdash +\partial_{\BEL}p$.

Proofs that if a rule is applicable/discarded in $D_{arg}$ at $P(n+1)$ then so
is the corresponding rule in $D_{\soc}$ at $P(n+1)$, are analogous to the
various cases studied for the inductive step of $+\partial_{X}$. Specifically,
we use $s_\alpha$ for rules captured by the quantifier in clause (2.3) of
$+\partial$ and $t_\beta$ for those in the scope of the quantifier of clause
(2.3.1).

It remains to argue that every applicable attack rule is defeated. By the
construction of the superiority relation, this statement is straightforward for
the rules which have a natural counterpart in $D_{arg}$, i.e., $t_{fl}>s_{fl}$,
$t_{Cvx}>s_{Cvx}$, $t_{Cvy}>s_{Cvx}$ when $\Aconf{Y}{X}\in \mathcal F$,
$t_{Cvy}>s_{Cvx}$.

Suppose $s_{\alpha}\in R[\flatx{x}{\non p}]$, with $\alpha\in\set{CvyCfx,
Cfbelx}$. Such a rule is defeated by $t_{Cvx}\in R[\flatx{x}{p}]$,
$t_{CvyCfx}\in R[\flatx{x}{p}]$, or $t_{Cfbelx}\in R[\flatx{x}{p}]$. All these
rule have the same counterpart rule, namely $t\in R^{\BEL}[p]$, what changes is
how $t$ is made applicable to challenge $s\in R^{\BEL}[\non p]$. Again, due to
construction of $>$ in Definition~\ref{def:theoryTransf}
$t_{\beta}>s_{\alpha}$, $\beta \in \set{Cvx, CvyCfx, Cfbelx}$, are so because
$t>_{\soc}s$.

The case when $s_{CfOI}\in R[\flatx{x}{\non p}]$ differs from the previous one
in that it can be defeated also by $t_{CfOI}$. Once more we have that
$t>^{sm}_{\soc}s$ by construction of $>$.

At last, we analyse the case when $D_{arg}\vdash +\partial \neg \flatx{x}{p}$.
The only rule that may fire to prove $\neg \flatx{x}{p}$ is $r_{neg-xp}$, which
is applicable whenever any rule $r_{dum-xp}$ is discarded at in $D_{arg}$ at
$P(n+1)$, due to conditions \ref{def2.dum-xp}--\ref{def2.neg-xp} and
construction of $>$ in Definition~\ref{def:theoryTransf}. That is the case if
$\flatx{x}{p}$ is $\partial$-rejected in $D_{arg}$ at $P(n)$. By inductive
hypothesis, $D_{\soc}\vdash -\partial_{X} p$ at $P(n)$, thus, by
Definition~\ref{def:provable} clause 3, $\neg Xp$ is $\partial$-provable in
$D_{\soc}$ at $P(n+1)$.

\paragraph{$P(n+1)= -\partial_{X} q$, $-\partial \Flatulo q$.} The main
reasoning follows straightforwardly from the case given that the proof
conditions for $-\partial_{X}$ and $-\partial$ are the strong negation of
$+\partial_{X}$ and $+\partial$, respectively. Clauses (1) and (2.1) of
$-\partial_{X}$ ($-\partial$) have already been proved in the inductive step of
$\pm\Delta_{X}$ ($\pm\Delta$), as well as clauses (2.2)-(2.3) in the inductive
step of $+\partial_{X}$ ($+\partial$).

By construction of the superiority relation given in
Definition~\ref{def:theoryTransf}, if $t\not >_{\soc}
s$, then no superiority relation may exist between any transformations of $t$
and $s$ in $D_{arg}$. 

Finally, concerning $P(n+1) = -\partial_{X}$, we must consider the case when
$\neg Xq$ is $\partial$-provable in $D_{\soc}$ at $P(n+1)$. This is the case
when clause 3. of Definition~\ref{def:provable} is satisfied, i.e., when
$D_{\soc}\vdash -\partial_{X}q$ at $P(n)$. By inductive hypothesis,
$D_{arg}\vdash -\partial \flatx{x}{q}$ at $P(n)$, making rules $r_{dum-xp}$
discarded in $D_{arg}$ at $P(n+1)$. Accordingly, $D_{arg}\vdash +\partial \non
xp$ at $P(n+2)$ and we conclude that rules $r_{neg-xp}$ prove $\neg
\flatx{x}{q}$ at $P(n+3)$.
\end{proof}

